\documentclass[a4paper,11pt]{article}
\usepackage{natbib}
\usepackage{amssymb,type1cm}
\usepackage[pdftex,bookmarks=true,bookmarksopen=true]{hyperref}

\title{The second-price auction solves \\ King Solomon's dilemma
\thanks{Preprint, Japanese Economic Review (2011)
\href{http://dx.doi.org/10.1111/j.1468-5876.2011.00543.x}{doi:10.1111/j.1468-5876.2011.00543.x}}
\thanks{The title is slightly rhetorical.
I would like to thank Takuma Wakayama, Yosuke Yasuda, and two referees for valuable comments.}}

\author{H.  Reiju Mihara\thanks{\emph{URL:} \url{http://econpapers.repec.org/RAS/pmi193.htm}
 (H.~R. Mihara).}\\ 
 Kagawa University Library, Takamatsu, 760-8525, Japan}

\date{July 2011}

\newcommand{\R}{\mathbb{R}}

\newcommand{\qed}{\enspace\enspace \vrule height 6pt width5pt
depth2pt}

\usepackage{theorem}

\newtheorem{theorem}{Theorem}
\newtheorem{prop}[theorem]{Proposition}

{\theorembodyfont{\normalfont}
\newtheorem{remark}{Remark} 
}

\newenvironment{proof}{\emph{Proof}.}{\qed\bigskip}

\begin{document} 

\maketitle 


\begin{abstract} 
The planner wants to give $k$ identical, indivisible objects to the top $k$ valuation agents at zero costs.
Each agent knows her own valuation of the object and whether it is among the top~$k$.
Modify the $(k+1)$st-price sealed-bid auction by introducing a small participation fee 
and the option not to participate in it.
This simple mechanism implements the desired outcome in iteratively undominated strategies.
Moreover, no pair of agents can profitably deviate from the equilibrium by coordinating their strategies or
bribing each other.

\emph{Journal of Economic Literature} Classifications:  C72, D61, D71, D82.

\end{abstract}

\pagebreak

\section{Introduction}

King Solomon's problem, when generalized to multiple units of an item, is described as follows: 
$k$~identical, indivisible objects are to be allocated among $n$ agents, where $k<n$.  
The objective of the ``planner'' (``auctioneer'') is to give the objects at no cost 
to the $k$ agents with the highest valuations.\footnote{There are many situations of this sort, 
including those mentioned in \citet[footnote~1]{glazer-m89}.}
I make the following informational assumptions:
First, the agents and the planner know that there is a gap greater than $\delta>0$
between the $k$th and the ($k+1$)st valuations (inequality (\ref{ineq2})).
Second, each agent knows not only her own valuation of the object but also whether she is among the top 
$k$ valuation agents.

For $k=1$, I propose the following variant of the second-price (sealed-bid) auction to solve the problem.
First, the agents need not participate in the auction if they do not want to.
Second, they may have to pay a small participation fee---a situation that arises
if the number of actual participants exceeds one (or the number $k$ of objects to be allocated, more generally).
In other words, I modify the second-price auction (Vickrey auction) by introducing 
the option not to participate in it and 
an arbitrarily small entry fee (any positive amount not greater than the gap~$\delta$).
This simple, intuitively appealing mechanism (modified auction) solves the problem.
Most likely, this mechanism, with entry fees, is close to what ordinary people would think of when they learn
the notion ``second-price auction,'' hence the title. 

The $(k+1)$st-price auction (instead of the second-price auction) for $k$~objects, similarly modified,
solves the generalized problem (Proposition~\ref{main}).  In other words, this two-stage mechanism
\emph{implements} the desired outcome in \emph{iteratively undominated strategies}
(obtained by one round of elimination of \emph{all} weakly dominated strategies,
 followed by two rounds of elimination of \emph{all} \emph{strictly} dominated ones).
In fact, only those top $k$ valuation agents choose to participate in an auction, which rules out
the need to hold an auction.

The reasoning behind this conclusion is straightforward.
In this auction, it is a weakly dominant strategy for each agent to bid her (true) valuation.
The planner can set an entry fee equal to $\delta>0$ so that the top $k$ valuation agents 
can profitably obtain the object by paying the $(k+1)$st price and the entry fee.
Then the other agents will not enter the auction, since they can only expect to pay the entry fee
without getting the object.
While the logic is simple, a careful argument specifying \emph{what information is available to which agent} is called for.
I state as explicitly as possible the informational assumptions on which each step of the argument is based.

\bigskip

Earlier contributions to King Solomon's problem, such as \citet{glazer-m89} and \citet{moore92}, 
consider the case of $k=1$ object and assume that each agent knows the \emph{other agents}' valuations, too.
Under this complete information assumption, they construct multi-stage mechanisms that implement the outcome
in subgame-perfect equilibrium.  Though those mechanisms consist of more than two stages, 
they have an appealing feature that only one agent moves at each stage.

More recently, assuming that each agent only knows her \emph{own} valuation 
as well as whether she is one of the top $k$ valuation agents,
several authors have constructed ingenious mechanisms that implement the outcome
in iteratively undominated strategies.
\citet{perry-r99} and \citet{olszewski03} construct mechanisms for $k=1$.
\citet{bag-s05} extend Olszewski's mechanism to any~$k$
 (they have also investigated the complete information setting).
\citet{qin-y09} propose an alternative mechanism for any~$k$.\footnote{%
Their solution concept is also one round of elimination of weakly dominated strategies,
followed by two rounds of elimination of strictly dominated ones.
I would like to thank Takuma Wakayama for pointing out an earlier version of their paper.
Most results in the present paper were obtained independently and made public
on the website of Social Science Research Network in August 2006. 
} %
Like mine, these mechanisms consist of two stages and lack the feature of only one agent moving at each stage.

As pointed out in most of these papers~\citep{glazer-m89,moore92,perry-r99}, an
auction itself does not solve the problem, since it involves a transfer of money.  
It is interesting to note that, given this fact, all the authors who deal with the incomplete information
settings propose a mechanism,
 of which a stage game is a modified version of the second-price auction.\footnote{%
Those dealing with the complete information settings \citep{glazer-m89,moore92} also propose auction-like mechanisms,
though the bidding protocols are different from the second-price auction.}
However, their modifications are fairly sophisticated, not appearing as straightforward as mine.
\citet{perry-r99} use a second-price \emph{all-pay} auction with the winner having an 
\emph{ex post option} to quit.  \citet{olszewski03} uses the second-price auction 
modified by adding an \emph{extra, non-constant (positive) payment from} the planner.
\citet{qin-y09} use a second-price auction with entry fees, 
where (for $n=2$, $k=1$, and $i\ne j$) 
$i$'s entry fee is determined as a function of $j$'s bid $b_j$ and of $i$'s \emph{guess} of $b_j$.
Their mechanism loses the advantage of the second-price auction 
that each agent need not guess others' bids or valuations.\footnote{ \label{qin-y}
The endogeneity of the fees is needed in their paper because they 
allow the case where the higher valuation and the lower valuation can be arbitrarily close:
in (\ref{ineq2}) of Section~\ref{framework}, they assume $\delta\ge 0$ instead of $\delta>0$.
See footnote~\ref{ols-delta}.
Qin and Yang assume that each agent is an expected utility maximizer,
who forms a subjective distribution of the other's valuation conditional on her own.
This assumption, which I do not make, is needed to obtain an optimal guess
in their paper.}

In Section~\ref{discussion}, I compare my mechanism with Olszewski's, 
which is one of the simplest in the literature.
Olszewski's mechanism requires the planner to subsidize the agents out of the equilibrium path.\footnote{%
In contrast, if my mechanism fails at the first stage, what comes after the entry fees are collected
 is just an ordinary second-price auction.
So the mechanism is particularly attractive as a compromise solution 
in situations where solution based on price is not too problematic but has not 
been used (because of some sort of stigma), 
such as assignment of parking spaces at some university campus.
}  %
As a result, it is vulnerable to
\emph{collusion between agents that bribe each other} to coordinate their strategies.
In fact, they can profitably deviate from the equilibrium without even manipulating their bids
 (Proposition~\ref{bribes1}).
Unlike Olszewski's (and unlike the second-price auction), my mechanism is not vulnerable
 to such collusion (Proposition~\ref{bribes2}).
 
\section{Framework}\label{framework}

We consider the problem~$\mathcal{P}_n^k$,
a multi-unit generalization of King Solomon's problem: 
$k$ identical, indivisible objects are to be allocated among $n$ agents, where $0<k<n$.  
The objective of the planner is to give the objects to the top $k$ 
valuation agents at zero monetary costs to the planner and the agents.

The framework is as follows:
Let $N=\{1, \ldots, n\}$ be the set of agents.
Fix a certain number $\delta>0$, which is known to everyone (i.e., all agents and the planner).
Fix a set $Q\subset \R^n$ of possible profiles of valuations of the object
such that every profile $(v_1, \ldots, v_n)$ in $Q$ contains at least $k$ positive components.
(The valuation by the planner is understood to be zero.)
At \emph{Stage}~$0$, God (Nature) announces a pair~$(v, H)$, where $v=(v_1, \ldots, v_n)\in Q$ 
is a profile of valuations and $H\subset N$ is a set consisting of $k$ agents
such that $i\in H$ implies $v_i\ge v_j$ for all $j\in L:=N\setminus H$.
\emph{While no one needs to know the set $Q$ itself, 
everyone (including the planner) knows the following condition imposed on~$Q$}: 
if $i\in H$ and $j\in L$, then $v_i> 0$ and\footnote{\label{profiles}
This assumption is naturally satisfied if, for example, $n=3$, $k=2$ and either 
(i)~$Q=U_1\times U_2 \times U_3$ for some pairwise disjoint finite sets $U_1$, $U_2$, and $U_3$ in $\R_+$
or
(ii)~for some disjoint closed intervals $U$, $V$ in $\R_+$ such that $u\in U$ and $u'\in V$ imply $u>u'$, 
we have $Q=(U\times U\times V)\cup(U\times V\times U)\cup (V\times U\times U)$.}
\begin{equation}\label{ineq2}
v_i-v_j>\delta.
\end{equation}
This condition implies that given $v\in Q$, we have $H=H_v$, where
$H_v$ is the (uniquely determined) set of $k$ agents with the highest valuations at~$v$.
Inequality~(\ref{ineq2}), indicating that the difference between the top $k$ valuations
and the others exceeds $\delta$, will serve as a ``word of wisdom'' that facilitates the construction of a
successful mechanism.

Each agent~$i$ observes her own type $\theta_i=(v_i, H(i))$, 
where $H(i)=1$ or~$0$, depending on $i\in H$ or not.\footnote{\label{bag:info}
The implementability result of \citet{bag-s05} is valid under this assumption, 
though they make a stronger assumption that each $i$~observes $(v_i, H)$.}  %
Let $\Theta_i$ 
be the set of possible types of~$i$, $\Theta_{-i} :=(\Theta_j)_{j\ne i}$, and $\theta_{-i} :=(\theta_j)_{j\ne i}$.
Given $\theta_i \in \Theta_i$ , agent $i$ has the set $\Theta_{-i}[\theta_i]$
of possible types of the other agents.\footnote{$\theta_{-i} \in \Theta_{-i}[\theta_i]$
iff $\theta_{-i} \in \Theta_{-i}$ and $(\theta_i,\theta_{-i})=((v_1,H_v(1)), \ldots, (v_n, H_v(n))$ for some $v\in Q$.
Since $i$ need not have an exact knowledge of~$Q$, she need not know the set $\Theta_{-i}[\theta_i]$
exactly, but our informational assumption about her knowledge of~$Q$ is sufficient to obtain the results.}
Each agent $i$'s payoff is $u_i((x_i,y_i),\theta_i)=v_i x_i + y_i$,
where $x_i\in \{0,1\}$ is the number of units of the object and $y_i\in \R$
the payment that $i$ receives.

The planner does not observe God's announcement.
It is common knowledge that the agents and the planner 
have the knowledge described here.

In the terminology of implementation theory, the problem~$\mathcal{P}_n^k$ is that of implementing the
(single-valued) \emph{choice function}~$f$ defined as follows: $f$ assigns an allocation
$f(v)=(x_i, y_i)_{i\in N}\in (\{0,1\}\times \R)^n$ to each profile $v\in Q$ of valuations,
which allocation is defined by $(x_i, y_i)=(1,0)$ if $i\in H_v$ and $(x_i, y_i)=(0,0)$ if $i\notin H_v$.

\section{The Solution}

The mechanism~$\mathcal{M}_n^k$ consists of two stages, Stage~1 followed by Stage~2. 
(We can regard it as a single-stage mechanism by considering its strategic form representation.
However, our solution concept---iteratively undominated strategies---seems 
more appealing if the mechanism is presented in an extensive form.)
I describe Stage~2 first.

\emph{Stage}~2 is the \emph{$(k+1)$st-price sealed-bid auction} for $k$ objects,
except that the planner collects the \emph{participation fee} $\delta>0$ from each participating agent.
There are at least $k+1$ agents participating in the auction and 
each participating agent~$i$ bids $b_i\in\R$.
Thus, \emph{any} bid $b_i\in \R$ (including those not corresponding to any profile in~$Q$) is allowed at this stage.
Rearrange the \emph{named bids}~$(b_i, i)$ according to the lexicographic order---first 
in terms of the value $b_i$ (highest bid first), second in terms of the agent name~$i$ (lowest number first).
Let $b^{k+1}$ be the $(k+1)$st bid (i.e., the first component of the $(k+1)$st named bid according to the above order)
by the participating agents.
The following is what agent~$i$ receives, as well as her payoff $u_i$: 
(a)~if $b_i$ is among the  $k$ highest bids (i.e., $(b_i, i)$ is among the first $k$ named bids according
to the lexicographic order), then $i$ gets the object but pays the $(k+1)$st bid and 
the participation fee, implying $u_i=v_i-b^{k+1}-\delta$;
(b)~otherwise, $i$ pays the participation fee, implying $u_i=-\delta$.

\emph{Stage}~1 is a simultaneous-move game
 in which the agents say either 
 ``auction'' (which means that she is willing to move on to Stage~2 and participate in an auction) 
 or ``no (auction).''
 More formally,  agent $i$ chooses a first-stage move $m_i\in \{1,0\}$, with
1 denoting ``auction'' and 0 ``no.''  
If at least $k+1$ agents say ``auction,'' then (only) those agents move on to Stage~2;
 the others get nothing.
 If less than $k+1$ agents say ``auction,''  then they get the object; the others get nothing.
  (If no agent says ``auction,'' then no agent gets anything.)

After playing Stage~1, each agent who said ``auction'' 
observes whether there were at least $k+1$ such agents.
This means that she knows whether an auction is to be held, though she does not know who are participating.
More formally, (given the realization of a $v\in Q$) each agent $i$ has just one 
\emph{information set} belonging to the second-stage,
 which set consists of all tuples $(m_1, \ldots, m_n)$ of first-stage moves
that contain 1's in the $i$th component and in at least $k$ others.
For example, if $n=3$ and $k=1$,
then agent 1's second-stage information set is $\{(1,1,1), (1,1,0), (1,0,1)\}$.
Without making any inferences,
she can only tell whether the result $(m_1, \ldots, m_n)$ of the first-stage belongs to her information set.
Under this assumption, agent $i$'s (global) strategy can be defined as a function
that maps each type $\theta_i=(v_i, H(i))\in \Theta_i$ to a message
$s_i=(m_i, b_i)\in S_i:= \{(\textrm{1 (``auction'')}, \textrm{0 (``no'')}\} \times \R$,
where $m_i$ is a first-stage move and $b_i$ is a bid that she will make if she participates in an auction.\footnote{%
One can make alternative assumptions about the information sets without affecting the result.
The argument will be similar, though the notation may become slightly more complex.
For example, one can assume that each agent observes the set of agents who said ``auction''
(this means that each agent observes the tuple of first-stage moves).
For $n=3$ and $k=1$, this implies that, say, 
agent 1 has the following three second-stage information sets:
$\{(1,1,1)\}$, $\{(1,1,0)\}$, and $\{(1,0,1)\}$.
Since her bid in Stage~2 can depend on which of these three has occurred,
her global strategy in this case is a function that maps each $\theta_1$
to $(m_1, b_1^{\{(1,1,1)\}}, b_1^{\{(1,1,0)\}}, b_1^{\{(1,0,1)\}})$,
where $b_1^I\in \R$ is her bid at an information set~$I$.} 
A message $s_i\in S_i$ is also referred to as a \emph{strategy} (available in the mechanism).
Let $g(s)$ be the outcome $(x_i,y_i)_{i\in N}$ of the mechanism~$\mathcal{M}_n^k$
when the messages are~$s$.

\bigskip

To make precise the statement that the mechanism above implements the desired outcome
 in iteratively undominated strategies, I introduce a few terms.

Let $S_i=\{1,0\}\times \R$ be the set of strategies (messages) for each~$i$.
Generalizing the solution concept in \citet{moulin79} and \citet[page~282]{perry-r99}
to the incomplete information setting, 
I say that $s=(s_1, \ldots, s_n)\in S_1\times \cdots \times S_n$ is a profile of \emph{iteratively (weakly) undominated}
strategies (I sometimes say that $s$ is an \emph{equilibrium}) at $v\in Q$ 
if $s\in S_1^T[\theta_1^v] \times \cdots \times S_n^T[\theta_n^v]$ for $\theta^v=(v_i,H_v(i))_{i\in N}$ and 
for some integer~$T$ (terminal round),
where the sequence 
\[
\langle (S_i^t[\theta_i], \{s_{-i}^t | \theta_i\}): \theta_i\in \Theta_i, i\in N, t\in \{0, \dots, T+1\} \rangle
\]
($S_i^t[\theta_i]$ is the set of $i$'s own strategies remaining after the $t$th round and
$\{s_{-i}^t | \theta_i\}$ is the set of the others' strategies that are [from the viewpoint of~$i$]
possibly remaining after the $t$th round)
is obtained by the following procedure:
for each~$i$ and $\theta_i$, $S_i^0[\theta_i] = S_i$, $\{s^{0}_{-i} | \theta_i\}=S_{-i}$, 
$S_i^T[\theta_i] = S_i^{T+1}[\theta_i]$, 
and for each \emph{round} $t\in \{1, \ldots, T+1\}$, 
$S^t_i [\theta_i]$ 
is the set of weakly undominated strategies in $S^{t-1}_i[\theta_i]$ at $\theta_i$ against
the strategies in $\{s^{t-1}_{-i} | \theta_i\}$,\footnote{%
That is, $s_i\in S^t_i [\theta_i]$ if $s_i\in S^{t-1}_i [\theta_i]$ and there is no $s'_i\in S^{t-1}_i [\theta_i]$ such that 
$u_i(g_i(s'_i,s_{-i}), \theta_i)\ge u_i(g_i(s_i,s_{-i}), \theta_i)$
for all $s_{-i} \in \{s^{t-1}_{-i} | \theta_i\}$, with strict inequality for some $s_{-i} \in \{s^{t-1}_{-i} | \theta_i\}$.}
and
$\{s^{t}_{-i} | \theta_i\} := \bigcup_{\tilde{\theta}_{-i}\in \Theta_{-i}[\theta_i]} \prod_{j\ne i} S^t_j[\tilde{\theta}_j]$.\footnote{
That is, $s_{-i}\in \{s^{t}_{-i} | \theta_i\}$ if for some $\tilde{\theta}_{-i} \in \Theta_{-i}[\theta_i]$,
we have $s_j\in S^t_j[\tilde{\theta}_j]$ for all $j \ne i$.
Obviously, $\{s^{t}_{-i} | \theta_i\} \subseteq \{s^{t-1}_{-i} | \theta_i\}$ for all $t$.}
In each round~$t$, 
$\langle (S_i^t[\theta_i], \{s_{-i}^t | \theta_i\}): \theta_i\in \Theta_i, i\in N \rangle$
is obtained from
$\langle (S_i^{t-1}[\theta_i], \{s_{-i}^{t-1} | \theta_i\}): \theta_i\in \Theta_i, i \in N \rangle$
that has been obtained.
Note that in each round~$t$, \emph{all} weakly dominated strategies in $S_i^{t-1}[\theta_i]$ against the strategies
in $\{s_{-i}^{t-1} | \theta_i\}$ are eliminated.

I say that the mechanism~$\mathcal{M}_n^k$ \emph{implements the choice function $f$ in iteratively
undominated strategies} if for each $v\in Q$,
the outcome $g(s)$ corresponding to \emph{any} remaining strategy profile~$s$
(i.e., $s$~is a profile of iteratively undominated strategies at $v$) is~$f(v)$.
There may be many remaining~$s$, but they must all yield the same outcome 
$f(v)=(x_i, y_i)_{i\in N}$ (allocation) defined in Section~\ref{framework}.

\begin{prop} \label{main}
The mechanism~$\mathcal{M}_n^k$ solves the problem~$\mathcal{P}_n^k$; that is, 
it implements the choice function~$f$ in iteratively undominated strategies.
Furthermore, the profiles of iteratively undominated strategies can be 
obtained by one round of elimination of \emph{all} weakly dominated strategies, 
followed by (at most) two rounds of elimination of \emph{all strictly} dominated ones.\end{prop}

\begin{proof}
Choose a $v\in Q$.  
We apply the elimination procedure three times and show that any remaining
strategy profile yields the outcome~$f(v)$.
Let $H=H_v$, $L=N\setminus H$ and $b^*_i=b^*_i(v_i,H(i))=v_i$ for all $i\in N$.

\emph{In the first round, each $i\in N$ eliminates all the strategies $(1,b_i)=(\textup{``auction''}, b_i)$ 
such that $b_i \neq v_i$.}  Indeed, it  is weakly dominated by $(1, b^*_i)$.
To see this, fix any $(m_j,b_j)_{j\ne i}$.
Then, depending on $m=(m_1, \ldots, m_k)$, we have two cases.
If $m$ contains at most $k$ 1's (i.e., if there are at most $k$ agents saying ``auction''),
then an auction is not held.  In this case, $i$ is indifferent between the two strategies $(1,b_i)$ and $(1,b^*_i)$
(in either case, $i$ gets the object).
If $m$ contains more than $k$ 1's, then an auction is held among the agents saying ``auction.''
Since the fees $\delta$ are independent of the agents' moves, the well-known result for the $(k+1)$st-price auction
(for fixed bidders)
implies that $b^*_i=v_i$ is the unique weakly dominant strategy for each $i$.%
\footnote{
The conclusion, say for~$i$, 
can be derived by fixing the bids of the other participating agents and then 
comparing $i$'s payoffs for bidding her valuation ($b^*_i=v_i$) and
for bidding something else, for each case: (a)~bidding $b^*_i$ is among the $k$ highest bids,
and (b)~otherwise.}  %
This implies that $(1,b^*_i)$ is at least as good as $(1,b_i)$ for $i$ and sometimes better.

The only remaining strategies in this round are $(1,b^*_i)$ and $(0,b_i)$, where $b_i\in\R$ is arbitrary.
We show that these cannot be eliminated in this round.
For $(1,b^*_i)$ to be eliminated, it has to be weakly dominated by $(0,b_i)$, which gives a zero payoff to~$i$.
But this is impossible (even for $i\in L$ such that $v_i<-\delta$)
since, for some $(m_j,b_j)_{j\ne i}$,
$(1,b^*_i)$ gives a greater payoff than $(0,b_i)$ to~$i$.  
(For example, let $m_j=1$ and $b_j=v_i-2\delta$ for all $j\ne i$.
Then $(1,b^*_i)$ gives $u_i=\delta>0$.)
For $(0,b_i)$ to be eliminated, it has to be weakly dominated by $(1,b^*_i)$.
But this is impossible since, for some $(m_j,b_j)_{j\ne i}$,
$(0,b_i)$ gives a greater payoff than $(1,b^*_i)$ to~$i$.
(For example, let $m_j=1$ and $b_j=v_i+\delta$ for all $j\ne i$.
Then $(1,b^*_i)$ gives $u_i= -\delta$.)

\emph{In the second round, each $i\in H$ eliminates all the strategies $(0,b_i)=(\textup{``no''}, b_i)$.}
In fact, it is \emph{strictly} dominated by $(1,b^*_i)$ against the others' possibly remaining strategies.
To see this, fix any $(m_j,b_j)_{j\ne i}$ such that $m_j=1$ implies $b_j=\tilde{v}_j$, 
where $\tilde{v}_j$ is a possible value of $v_j$ from $i$'s viewpoint.
Then, depending on $m_{-i}=(m_j)_{j\ne i}$, we have two cases.
If $m_{-i}$ contains less than $k$ 1's (i.e., if less than $k$ other agents say ``auction''),
then an auction is not held in any case.
So $(1,b^*_i)$ is better than $(0,b_i)$ for $i$ since the former gives 
her the object with a payoff of $v_i> 0$ (since $i\in H$), while the latter gives her a zero payoff.
If $m_{-i}$ contains at least $k$ 1's, then we can show as follows that $(1,b^*_i)$ is better than $(0,b_i)$ for $i$.
By choosing $(1,b^*_i)$, $i$ can participate
in an auction, which is held.  In this case, $i$ knows that she will be among the $k$ highest bidders.
Her payoff from the auction is $u_i=v_i-b^{k+1}-\delta$, which depends on $b^{k+1}$ yet to be known.
But she can deduce that if $\tilde{v}_j=b_j=b^{k+1}$, 
then $\tilde{H}(j)=0$ (i.e., $j\in \tilde{L}$)
(if $\tilde{H}(j)=1$, then $\tilde{v}_j$ is among the top $k$ bids).
Then, by (\ref{ineq2}), $i\in H$ and $j\in \tilde{L}$ implies that
$u_i=v_i-\tilde{v}_j-\delta>0$.
Note that at the end of the second round, $i\in H$ has only one remaining strategy.

In the second round, $i\in L$ cannot eliminate any $(0,b_i)$.
For $(0,b_i)$ to be eliminated, it has to be weakly dominated by $(1,b^*_i)$.
But this is impossible since, for some $(m_j,b_j)_{j\ne i}$ such that $m_j=1$ implies $b_j=\tilde{v}_j$,
$(0,b_i)$ gives a greater payoff than $(1,b^*_i)$ to~$i$.
(For example, let $m_j=1$ and $b_j=v_j$ for all $j\ne i$.)

On the other hand, $i\in L$ may eliminate $(1,b^*_i)$ in this round, depending on her knowledge of~$Q$.
For example, if $i\in L$ knows that $v_i\le 0$ and $v_i-\tilde{v}_j\le \delta$ for all 
$\tilde{v}_j$ for $j\ne i$ such that $(v_i,\tilde{v}_{-i}) \in Q$,
then the strategy $(1,b^*_i)$ gives her a payoff of either 
$v_i\le 0$ or $v_i-b^{k+1}-\delta \le 0$ or $-\delta<0$.\footnote{%
This possibility can be ignored if we assume that $v_i>0$ for all $i\in N$ or
that $Q\subset \R^n$ is not bounded below in any dimension.
Also, if $(1,b^*_i)$ is only weakly (not strictly) dominated, then one should eliminate it
if he follows the procedure strictly.  Instead, one can go on to the next round
and eliminate it, which is now strictly dominated, thus obtaining the same result.}
 
 \emph{In the third round, $i\in L$ eliminates the strategy $(1,b^*_i)=(\textup{``auction''}, b_i)$},
if she has not done so in the second round. 
In fact, it is \emph{strictly} dominated by $(0,b_i)$ against the others' possibly remaining strategies.
To see this, fix any $(m_j,b_j)_{j\ne i}$ such that if $m_j=1$, then $b_j=\tilde{v}_j$ and if 
$\tilde{H}(j)=1$, then $m_j=1$.
Since $i$ knows that the agents $j\in \tilde{H}$ will choose ``auction'' and  bid $b_j=\tilde{v}_j$,
she knows that if she says ``auction,'' an auction is held and she gets the payoff of 
$u_i=-\delta<0$ (she cannot get the object because she will not be among the $k$ highest bidders).

At this point in the elimination process, the remaining strategies $(m_i, b_i)$ are such that 
$(m_i, b_i)=(1, v_i)$ if $i\in H$ and $m_i=0$ if $i\in L$.
It is easy to see that any profile of such strategies yields the same outcome
(hence there are no more strategies to be eliminated in the fourth round).
Indeed, since only those agents in~$H$ say ``auction'' and there are exactly $k$ such agents, an auction will not be held.
So each agent in $H$ gets the object and each agent in $L$ gets nothing.\end{proof}

\begin{remark}
It appears that the conclusion of Proposition~\ref{main} is no longer true if we consider
(instead of the iteratively undominated strategies, where \emph{all} weakly dominated strategies
are eliminated in each round)
the strategies that remain under
 \emph{different procedures for eliminating weakly dominated strategies}.
For example, consider the classical two-agent case ($n=2$, $k=1$, $v_1>0$, $v_2>0$).
Suppose all weakly dominated strategies are eliminated only for the higher valuation agent $i\in H$ in the first round.
Thus, the remaining strategies for $i\in H$ are $(1,b^*_i) = (1, v_i)$ and $(0,b_i)$, where $b_i\in\R$ is arbitrary.
Then in the second round, the lower valuation agent $j\in L$ cannot eliminate
$(1, b_j) = (1, v_j+\delta)$.
To see this, note that if $(1, b_j)$ is weakly dominated by some strategy, it has to be
weakly dominated by $(1, b^*_j)=(1, v_j)$.
But it is easy to see that (since $\tilde{v}_i-v_j>\delta$ 
implies $b_j = v_j+\delta < \tilde{v}_i = b^*_i$)
$(1, b_j)$ and $(1, b^*_j)$ give the same payoffs for any strategy for~$i$.
Since $(1,b_j)$ was not eliminated in the second round, we cannot conclude that 
$i\in H$ eliminates all the strategies $(0,b_i)$ in the third round.
This is because she might obtain a negative payoff by participating in the auction.
The argument of the proof fails for this procedure.\end{remark}

\section{Discussion}\label{discussion}

It would be of some interest to compare the mechanism~$\mathcal{M}_n^k$ with
those in the literature~\citep{perry-r99, olszewski03, bag-s05,qin-y09}
dealing with the incomplete information environments.
Of those mechanisms, I focus on Olszewski's since it is simpler than
Perry and Reny's.  
Also,  Bag and Sabourian's mechanism for the incomplete information setting 
is an extension of Olszewski's, not an alternative to it.
Qin and Yang's mechanism performs just like mine, if we ignore the complexity of making guesses (see footnote~\ref{qin-y}).

I focus on the classical case of Solomon's problem in this section: $n=2$, $k=1$, $v_1>0$, $v_2>0$,
and for $i\in H$ (the higher-valuation agent) and $j\in L$ (the lower-valuation agent),
$v_i-v_j> \delta>0$.\footnote{\label{ols-delta}%
Olszewski constructs another mechanism that solves the problem for
 $\delta=0$ (the case where the higher valuation and the lower valuation can be arbitrarily close).
My mechanism fails to solve such a problem:  if $\delta=0$, the
lower valuation agent's strategy ``no'' is weakly dominated in the second round of elimination 
if she has a positive valuation.
Hence the usual second-price auction (no participation fees) will be held.}  %
Note that the planner can use arbitrarily small $\delta$ (because 
if $\delta>0$ satisfies the inequality, then so does any positive $\delta'\le \delta$).

Olszewski's mechanism works as follows: 
In Stage~1, the two agents say ``hers'' (corresponding to ``auction'' in this paper)
or ``mine'' (``no (auction)'') simultaneously.  If both say ``hers,'' then they \emph{move on to Stage~2}.
If only one says ``hers,'' then the agent who says ``mine'' gets the object. 
If both say ``mine,'' then both get \emph{nothing}.
Stage~2 is a modified second-price auction (modified such that
each agent pays the entrance fee $\delta$ but receives the other's bid):
if $b_i>b_j$, then $u_i=v_i-\delta$ and $u_j=b_i-\delta$.

Table~\ref{comparison} compares the payoffs for the two mechanisms,
 assuming $b_i>b_j$ and $i$ is the row player.

\begin{table}[htdp]
\begin{center}
\begin{tabular}{ccccccc}
 & ``hers''  & ``mine'' & \hspace{3mm}  & & ``auction''  & ``no'' \\ \cline{2-3} \cline{6-7}
``hers'' & $v_i-\delta, b_i-\delta$ & $0, v_j$ & & ``auction'' &  $v_i-b_j-\delta, -\delta$ &  $v_i, 0$  \\
``mine'' & $v_i, 0$ & $0, 0$ & & ``no'' & $0, v_j$ & $0, 0$ 
\end{tabular}
\end{center}
\caption{Payoffs for Olszewski's mechanism (left) and mine (right).}
\label{comparison}
\end{table}%

It is a weakly dominant strategy for each~$i$ to play $b_i=v_i$ in Stage~2.
The other strategies are eliminated in the first round of elimination of weakly dominated strategies.
Olszewski's mechanism requires another round:
``hers'' is a weakly (but not strictly) dominated strategy for the higher-valuation agent
and ``mine'' is one for the lower-valuation agent (if $i\in H$ and $j\in L$, then
$v_j<u_i=u_j=v_i-\delta<v_i$).
My mechanism requires two more rounds.  
But those strategies to be eliminated in the second and the third rounds are
 \emph{strictly} dominated.

Olszewski's mechanism relies on the availability of transfer from the planner 
out of the equilibrium path.\footnote{\label{subsidy}
The total amount received by the agents in Stage~2 is 
$-\delta+(b_i-\delta)=b_i-2\delta$.
If we require this value to be non-positive,  even if we assume $b_i=v_i$, we have
$v_i\le 2\delta < 2v_i-2v_j$, implying the inequality~$v_i>2v_j$, not likely in many situations.}
The reliance on subsidies from outside means that the agents are less
likely to find an outsider (planner) who is willing to adopt this mechanism.
In contrast, the total amount received by the agents in Stage~2 of my mechanism
 is negative ($(-b_j-\delta)-\delta=-b_j-2\delta=-v_j-2\delta<0$).
 
 \bigskip

I next consider the possibility of monetary transfers (not described by the mechanisms) between the agents.
I assume that the agents can bribe each other to coordinate their strategies.
Let $u_i(s)$ be $i$'s payoff from a mechanism, where $s=(s_i,s_j,s_{-ij})$ and $s_{-ij}=(s_k)_{k\notin\{i,j\}}$.
We say that a strategy profile $s=(s_i)$ is \emph{stable against pairwise (coalitional) deviations with transferable utility}
 if the following condition is violated:
there are two agents $i$, $j$, their strategies $s'_i$, $s'_j$, and a bribe $t\in\R$ such that 
$u'_i:=u_i(s'_i,s'_j,s_{-ij})+t >u_i(s)$
and
$u'_j:=u_j(s'_i,s'_j,s_{-ij})-t >u_j(s)$.

\begin{prop} \label{bribes1}
Suppose that the total amount received by the agents
in Stage~2 of Olszewski's mechanism is positive, assuming $i\in H$ bids $b_i=v_i$.
Then its equilibrium is not stable against pairwise deviations with transferable utility
even if the agents bid their valuations in Stage~2.
\end{prop}

\begin{proof}
Consider the strategies such that both agents say ``hers'' and bid their valuations.
The total amount subsidized is $b_i-2\delta=v_i-2\delta>0$ by assumption. 
Find an $\epsilon>0$ such that $b_i-2\delta-\epsilon >0$.
Consider a bribe $\delta+\epsilon$ from $j\in L$ to $i\in H$.
Since $b_i>b_j$, the resulting payoffs are:
$u'_i=v_i-\delta+\delta+\epsilon=v_i+\epsilon>v_i$; 
$u'_j=b_i-\delta-\delta-\epsilon =b_i-2\delta-\epsilon>0$.\end{proof}

Note that if the agents are not restricted to bidding their (true) valuations, they can achieve arbitrarily
large payoffs,\footnote{
For any $\bar{u}_i>0$ and $\bar{u}_j>0$, fix a small $b_j$, 
find a bribe $t\in\R$ such that $u'_i=v_i-\delta+t >\bar{u}_i$,
and find a $b_i$ such that $u'_j=b_i-\delta-t>\bar{u}_j$.}
though the availability of subsidies then becomes questionable.

In contrast, my mechanism works better against bribes.  I present the result 
for a more general case of $n$~agents and $k$~objects; it includes the classical case.

\begin{prop}\label{bribes2}
Suppose that each individual has a positive valuation and 
is prohibited from submitting a negative bid: $v_i>0$ and $b_i\ge 0$ for each $i$.
Then the equilibrium of the mechanism~$\mathcal{M}_n^k$ is stable against pairwise deviations with transferable utility.\end{prop}

\begin{proof}
Let $s$ be an equilibrium and suppose it not stable.
Then there are agents $i$, $j$, strategies $s'_i$, $s'_j$, and a bribe $t$ such that 
$u'_i:=u_i(s'_i,s'_j,s_{-ij})+t >u_i(s)$ and $u'_j:=u_j(s'_i,s'_j,s_{-ij})-t >u_j(s)$.
We have
\begin{equation} \label{better-off}
u'_i+u'_j = u_i(s'_i,s'_j,s_{-ij})+u_j(s'_i,s'_j,s_{-ij}) > u_i(s)+u_j(s).
\end{equation}

Suppose $i$, $j\in H$.  Then $u_i(s)+u_j(s)=v_i+v_j$.
Inequality  (\ref{better-off}) cannot be satisfied since $u_i(s')\leq v_i$ and $u_j(s')\leq v_j$ for any $s'$.

Suppose $i$, $j\in L$.  If both say ``no,'' they cannot meet inequality (\ref{better-off}).
So, suppose that $i$ says ``auction,'' in which case she is worse off (regardless of whether she gets the object),
unless she receives a sufficiently large bribe $t>0$.
Then $j$, who pays the bribe, is worse off (whether she participates in the auction), 
violating $u'_j>u_j(s)$.

It follows that $i\in H$ and $j\in L$ without loss of generality.

(i)~Suppose $i$ says ``auction'' and $j$ says ``no.''
Then $u'_i=v_i+t>u_i(s)=v_i$ implies $u'_j=0-t<0=u_j(s)$, a contradiction.

(ii)~Suppose $i$ says ``no'' and $j$ says ``auction.''
Then $u'_i=0+t>u_i(s)=v_i$ implies $u'_j=v_j-t<v_j-v_i< -\delta< 0=u_j(s)$, a contradiction.

(iii)~Suppose $i$ says ``no'' and $j$ says ``no.''
Then $u'_i+u'_j=0$ and $u_i(s)+u_j(s)=v_i$, violating (\ref{better-off}).

(iv)~Suppose $i$ says ``auction'' and $j$ says ``auction.''
If $j$ gets the object, (\ref{better-off}) implies that $u'_i+u'_j=v_j-b^{k+1}-2\delta> v_i$, where $b^{k+1}$ 
is the $(k+1)$st highest bid.  Then $-b^{k+1}-2\delta> v_i-v_j>\delta$, implying $-b^{k+1}> 3\delta>0$,
contradicting the assumption that bids are nonnegative.  The case where $i$ gets the object is easier.\end{proof}





\ifx\undefined\bysame
\newcommand{\bysame}{\hskip.3em \leavevmode\rule[.5ex]{3em}{.3pt}\hskip0.5em}
\fi

\end{document}